\documentclass[a4paper,11pt]{article}

\usepackage{amsmath}
\usepackage{amsthm}
\usepackage{amssymb}

\usepackage{algorithmic}
\usepackage{algorithm}
\usepackage{xcolor}
\usepackage{graphicx}

\usepackage{cite}
\usepackage{paralist}

\usepackage[left=0.8in,top=1.0in,right=0.8in,bottom=1.0in,nohead]{geometry}

\theoremstyle{definition}

\newtheorem{definition}{Definition}
\newtheorem{theorem}{Theorem}

\newtheorem{lemma}[theorem]{Lemma}

\graphicspath{{figs/}}

\def\MC#1{{\mathcal #1}}
\def\MBB#1{{\mathbb #1}}
\def\MB#1{{\mathbf #1}}

\newcommand{\BIGLR}[3]{{\left#1#3\right#2}}
\newcommand{\BIGP}[1]{{\BIGLR{(}{)}{#1}}}
\newcommand{\BIGCP}[1]{{\BIGLR{[}{]}{#1}}}
\newcommand{\BIGBP}[1]{{\BIGLR{\{}{\}}{#1}}}
\newcommand{\CEIL}[1]{{\BIGLR{\lceil}{\rceil}{#1}}}
\newcommand{\FLOOR}[1]{{\BIGLR{\lfloor}{\rfloor}{#1}}}
\newcommand{\BIGC}[1]{{\BIGLR{|}{|}{#1}}}

\long\def\longdelete#1{}

\title{Competitive Design and Analysis for Machine-Minimizing Job Scheduling Problem
\thanks{This work was supported in part by National Science Council (NSC), Taiwan, under Grants NSC99-2911-I-002-055-2, NSC98-2221-E-001-007-MY3, and Karlsruhe House of Young Scientists (KHYS), KIT, Germany, under a Grant of Visiting Researcher Scholarship.}}

\author{Mong-Jen~Kao$^2$ \and Jian-Jia~Chen$^1$ 
\and Ignaz~Rutter$^1$ \and Dorothea~Wagner$^1$ \\
\\
$^1$ Faculty for Informatics, Karlsruhe Institute of Technology (KIT), Germany. \\
$^2$ Research Center for Infor. Technology Innovation, Academia Sinica, Taiwan. \\
\textit{Email: mong@citi.sinica.edu.tw, j.chen@kit.edu, 
rutter@kit.edu, dorothea.wagner@kit.edu}
}

\date{}

\begin{document}

\maketitle


\begin{abstract}
We explore the machine-minimizing job scheduling problem, which has a rich history in the line of research, under an online setting.
We consider systems with arbitrary job arrival times, arbitrary job deadlines, and unit job execution time. For this problem, we present a lower bound $2.09$ on the competitive factor of \emph{any} online algorithms, followed by designing a $5.2$-competitive online algorithm.
We also point out a false claim made in an existing paper of Shi and Ye regarding a further restricted case of the considered problem.
To the best of our knowledge, what we present is the first concrete result concerning online machine-minimizing job scheduling with arbitrary job arrival times and deadlines.
%
\end{abstract}

%


\section{Introduction}

Scheduling jobs with interval constraints is one of the most well-known models in classical scheduling theory that provides an elegant formulation for numerous applications and which also has a rich history in the line of research that goes back to the 1950s. For example, assembly line placement of circuit boards~\cite{Crama98theassembly,Spieksma99}, time-constrained communication scheduling~\cite{Adler:1998:STC:277651.277693}, adaptive rate-controlled scheduling for multimedia applications~\cite{Yau:1997:ARS:244130.244177,Rajugopal:1998:ARC:284583.284590}, etc.

\smallskip

In the basic framework, we are given a set of jobs, each associated with a set of time intervals during which it can be scheduled. 
Scheduling a job means selecting one of its associated time interval. The goal is to schedule all the jobs on a minimum number of machines such that no two jobs assigned to the same machine overlap in time.
%
Two variations have been considered in the literature, differing in the way how the time intervals of the jobs are specified. In the \emph{discrete machine minimization}, the time intervals are listed explicitly as the input, while in the \emph{continuous} version, the set of time intervals for each job is specified by a release time, a deadline, and an execution time.

\smallskip

In terms of problem complexity, it is known that deciding whether one machine suffices to schedule all the jobs is already strongly NP-complete~\cite{Garey:1979:CIG:578533}.
Raghavan and Thompson~\cite{Raghavan:1987:RRT:45291.45296} gave an $O\BIGP{{\log n} / {\log\log n}}$-approximation via randomized rounding of linear programs for both versions. This result is also the best known approximation for the discrete version. An $\Omega\BIGP{\log\log n}$
lower-bound on the approximation ratio is given by Chuzhoy and Naor~\cite{Chuzhoy:2006:NHR:1183907.1183908}.
For the continuous machine minimization, Chuzhoy et al.~\cite{Chuzhoy:2004:MMS:1032645.1033163} improved the factor to $O\BIGP{\sqrt{\log n}}$. When the number of machines used by the optimal schedule is small, they provided an $O\BIGP{k^2}$-approximation, where $k$ is the number of machines used by the optimal schedule.
Recently, Chuzhoy et al.~\cite{Chuzhoy:2009:RMJ:1616497.1616504} further improved their previous result to a (large) constant.

\smallskip

In addition,
results regarding special constraints have been proposed as well. Cieliebak et al.~\cite{springerlink:10.1007/1-4020-8141-318} studied the situation when the lengths of the time intervals during which the jobs can be scheduled are small. Several exact algorithms and hardness results were presented.
Yu and Zhang~\cite{Yu200997} considered two special cases. When the jobs have equal release times, they provided a $2$-approximation. When the jobs have equal execution time, they showed that the classical greedy best-fit algorithm achieves a $6$-approximation.

\smallskip

%

From the perspective of utilization-enhancing, a problem that can be seen as dual to machine minimization is the \emph{throughput maximization} problem, whose goal is to maximize the number of jobs that can be scheduled on a single machine.
%
Chuzhoy et al.~\cite{Chuzhoy:2006:AAJ:1235273.1235278} provided an $O\BIGP{\frac{e}{e-1}+\epsilon}$-approximation for any $\epsilon > 0$ for both discrete and continuous settings, where $e$ is the Euler's number. Spieksma~\cite{Spieksma99} proved that the discrete version of this problem is MAX-SNP hard, even when the set of time intervals for each job has cardinality two.

\smallskip

Several natural generalizations of this problem have been considered. Bar-Noy et al.~\cite{Bar-Noy:2002:ATM:586839.586857} considered the weighted throughput maximization problem, in which the objective is to maximize the weighted throughput for a set of weighted jobs, and presented a $2$-approximation.
Furthermore, when multiple jobs are allowed to share the time-frame of the same machine, i.e., the concept of context switch is introduced to enhance the throughput, Bar-Noy et al.~\cite{Bar-Noy:2001:UAA:502102.502107} presented a $5$-approximation and a $\BIGP{\frac{2e-1}{e-1}+\epsilon}$-approximation for both weighted and unweighted versions.
When the set of time intervals for each job has cardinality one, i.e., only job selection is taken into consideration to maximize the weighted throughput, Calinescu et al.~\cite{Calinescu:2011:IAA:2000807.2000816} presented a $\BIGP{2+\epsilon}$-approximation while Bansal et al.~\cite{Bansal:2006:QUF:1132516.1132617} presented a quasi-PTAS.

%


\paragraph{Our Focus and Contribution.}

In this paper, we explore the \emph{continuous} machine-minimizing job scheduling problem under an online setting and presents novel competitive analysis for this problem. We consider a real-time system in which we do not have prior knowledge on the arrival of a job until it arrives to the system, and the scheduling decisions have to be made online.
%
%
As an initial step to exploring the general problem complexity, we consider the case for which all the jobs have \emph{unit} execution time.
In particular, we provide for this problem:

\begin{itemize}
	\item
		a $2.09$ lower bound on the competitive factor of \emph{any} online algorithm, and
		
		\smallskip
		
	\item
		a $5.2$ competitive online algorithm.
\end{itemize}

To the best of our knowledge, this is the first result presented under the concept of real-time machine minimization with arbitrary job arrival times and deadlines.

\smallskip

We would also like to point out a major flaw in a previous paper~\cite{Shi:2008:OBP:1328331.1328337} in which the authors claimed to have an optimal $2$-competitive algorithm for a restricted case where the jobs have a universal deadline.
In fact, our lower bound proof is built exactly under this restricted case,
thereby showing that even when the jobs have a universal deadline, any \emph{feasible} online algorithm has a competitive factor no less than $2.09$.

%


\section{Notations, Problem Model, and Preliminary Statements}

This section describes the job scheduling model adopted in this paper, followed by a formal problem definition and several technical lemmas used throughout this paper.

\subsection{Job Model} 

We consider a set of real-time jobs, arriving to the system
dynamically. When a job $j$ arrives to the system, say, at time $t$, its
\emph{arrival time} $a_j$ is defined to be $t$ and the job is put into the \emph{ready queue}. 
The \emph{absolute deadline}, or, deadline for simplicity, for which $j$ must finish its execution is denoted by $d_j$.
The amount of time $j$ requires to finish its execution, also called the \emph{execution time} of $j$, is denoted by $c_j$.
We consider systems with \emph{discretized timing line} and \emph{unit jobs}, i.e., 
$a_j$ and $d_j$ are non-negative integers, and $c_j = 1$.
When a job finishes, it is removed from the ready queue.
For notational brevity, for a job $j$, we implicitly use a
pair $j = \BIGP{a_j,d_j}$ to denote the corresponding properties.
%
%

%

\subsection{Job Scheduling}

A schedule $\MB{S}$ for a set of jobs $\MC{J}$ is to decide for each job $j \in \MC{J}$ the time at which $j$ starts its execution. $\MB{S}$ is said to be \emph{feasible} if each job starts its execution no earlier than its arrival and has its execution finished at its deadline.
Moreover, we say that a schedule $\MB{S}$ follows the earliest-deadline-first (EDF) principle if 
whenever there are multiple choices on the jobs to schedule, it always gives the highest priority to the one with earliest deadline.

Let $\#_\MB{S}(t)$ be the number of jobs which are 
scheduled for execution at time $t$ in schedule $\MB{S}$. The number of machines $\MB{S}$ requires to finish the execution of 
the entire job set, denoted by $M(\MB{S})$, is then $\max_{t \ge 0}\#_\MB{S}(t)$.
For any $0\le \ell < r$, let 
$$\MC{J}(\ell,r) = \BIGC{\BIGBP{j:j\in \MC{J}, \: \ell\le a_j, \: d_j \le r}}$$ 
denote the total amount of workload, i.e., the total number of jobs due to the unit execution time of the jobs in our setting, that arrives and has to be done within the time interval $\BIGCP{\ell,r}$.
The following lemma provides a characterization of a feasible EDF schedule for any job set.

\begin{lemma}
\label{lemma-feasibility-characterization}
For any set $\MC{J}$ of unit jobs, a schedule $\MB{S}$ following the earliest-deadline-first principle is feasible if and only if for any $0\le \ell < r$,
$$\sum_{\ell\le t<r}\#_\MB{S}(t) \ge \MC{J}(\ell,r).$$
\end{lemma}

\begin{proof}
The "only if" part is easy to see. For any feasible schedule $\MB{S}$ and any $0 \le \ell < r$, the amount of workload which arrives and has to be done with the time interval $\BIGCP{\ell,r}$ must have been scheduled, and, can only be scheduled within the interval $\BIGCP{\ell,r}$.
Since $\MB{S}$ is feasible, we get $\sum_{\ell\le t<r}\#_\MB{S}(t) \ge \MC{J}(\ell,r)$ for any $0\le \ell < r$.

\medskip

Below we prove the "if" part.
Assume for contradiction that a schedule $\MB{S}$ following the earliest-deadline-first principle fails to be feasible when we have $\sum_{\ell\le t<r}\#_\MB{S}(t) \ge \MC{J}(\ell,r)$, for all $0\le \ell < r$.
Let $j = (a_j, b_j)$ be one of the jobs which misses their deadlines in $\MB{S}$, and $\MC{S}(a_j, b_j)$ be the set of jobs scheduled to be executed during the time interval $\BIGCP{a_j, b_j}$ in $\MB{S}$.
Since $\MB{S}$ follows the EDF principle, all the jobs in $\MC{S}(a_j, b_j)$ have the deadlines no later than $b_j$. By the assumption, we know that, $$\sum_{a_j \le t < b_j}\#_\MB{S}(t) \ge \MC{J}(a_j, b_j).$$
Therefore at least one job in $\MC{S}(a_j, b_j)$ has arrival time earlier than $a_j$. Let $j^\prime = (a_{j^\prime}, b_{j^\prime})$ be one of such job. By exactly the same argument, we know that all the jobs in $\MC{S}(a_{j^\prime}, a_j)$, which is the set of jobs scheduled to be executed during $\BIGCP{a_{j^\prime}, a_j}$, have the deadlines no later than $b_{j^\prime}$. Since $$\sum_{a_{j^\prime} \le t < a_j}\#_\MB{S}(t) \ge \MC{J}(a_{j^\prime}, a_j),$$
at least one job in $\MC{S}(a_{j^\prime}, a_j)$ has arrival time earlier than $a_{j^\prime}$. Let $j^{\prime\prime} = (a_{j^{\prime\prime}}, b_{j^{\prime\prime}})$ be one of such job. 
Note that the arrival times of $j$, $j^\prime$, and $j^{\prime\prime}$, are strictly decreasing, i.e., $a_{j^{\prime\prime}} < a_{j^\prime} < a_j$, and the deadlines are non-increasing, i.e., $b_{j^{\prime\prime}} \le b_{j^\prime} \le b_j$.
Therefore, each time we apply this argument, we get an interval whose left-end is strictly smaller than the previous one, while the right-end is no larger than the previous one. By assumption, this gives another job that leads to another interval with the same property.
Hence, by repeating the above argument, since number of jobs which have arrived to the system up to \emph{any} moment is finite, eventually, we get a contradiction to the assumption that $\sum_{\ell\le t<r}\#_\MB{S}(t) \ge \MC{J}(\ell,r)$ for all $0\le \ell < r$, which implies that the schedule $\MB{S}$ must be feasible.
\qed
\end{proof}

In the \emph{offline machine-minimizing job scheduling} problem, we 
wish to find a schedule $\MB{S}_\MC{J}$ for a given set of jobs $\MC{J}$ such that $M(\MB{S}_\MC{J})$ is minimized. 
For a better depiction of this notion, for any $0\le \ell < r$, let 
$$\rho\BIGP{\MC{J},\ell,r} = \frac{\MC{J}(\ell,r)}{r-\ell}$$ denote the density of workload $\MC{J}(\ell,r)$, and 
%
let $OPT\BIGP{\MC{J}}$ denote the number of machines required by an optimal schedule for $\MC{J}$.
The following lemma shows that, when the jobs have unit execution time, there is a direct link between $OPT\BIGP{\MC{J}}$ and the density of workload.

%

\begin{lemma}
\label{lemma-opt-offline-density}
For any set of jobs $\MC{J}$,
we have
$OPT\BIGP{\MC{J}} = \CEIL{\max_{0\le\ell<r}\rho\BIGP{\MC{J},\ell,r}}.$
\end{lemma}

\begin{proof}
By Lemma~\ref{lemma-feasibility-characterization}, we have $\sum_{\ell\le t<r}\#_\MB{S}(t) \ge \MC{J}(\ell,r)$, for any specific pair $(\ell, r)$, $0\le \ell < r$.
Dividing both side of the inequality by $(r-\ell)$, we get 
$$\frac{1}{r-\ell}\cdot\sum_{\ell\le t<r}\#_\MB{S}(t) \ge \frac{\MC{J}(\ell,r)}{r-\ell} = \rho\BIGP{\MC{J}, \ell, r}.$$
From the definition of $M(\MB{S})$, we also have $$M(\MB{S}) \ge \max_{\ell \le t < r}\#_\MB{S}(t) \ge \CEIL{\frac{1}{r-\ell}\cdot\sum_{\ell\le t<r}\#_\MB{S}(t)},$$
where the ceiling comes from the fact that the number of machines is an integer.
This implies that $M(\MB{S}) \ge \CEIL{\rho\BIGP{\MC{J},\ell,r}}$. Since this holds for all $0\le \ell < r$, we get $M(\MB{S}) \ge \CEIL{\max_{0 \le \ell < r}\rho\BIGP{\MC{J},\ell,r}}$, for any feasible schedule $\MB{S}$, including $OPT(\MC{J})$.

To see that the inequality becomes equality for $OPT(\MC{J})$. Consider the schedule $\MB{S}^\prime$ which follows the earliest-deadline-first principle and which uses exactly $\CEIL{\max_{0 \le \ell < r}\rho\BIGP{\MC{J},\ell,r}}$ machines at all time. We have
$$\sum_{\ell \le t < r}\#_{\MB{S}^\prime}(t) \ge \BIGP{r-\ell}\cdot\CEIL{\rho\BIGP{\MC{J},\ell,r}} \ge \MC{J}(\ell,r),$$
for all $0\le \ell < r$. This implies $\MB{S}^\prime$ is feasible by Lemma~\ref{lemma-feasibility-characterization}. Since $\MB{S}^\prime$ uses the minimum number of machines, it is one of the optimal schedule. This proves the lemma.
\qed
\end{proof}

\subsection{Online Job Scheduling} 

We consider the case where the jobs are arriving in an online setting, i.e., at any time $t$, we only know the job arrivals up to time $t$,
and the scheduling decisions have to be made without prior knowledge on future job arrivals.
To be more precise, let 
$\MC{J}(t) = \BIGBP{j \colon j\in \MC{J}, a_j \le t}$
be the subset of $\MC{J}$ which contains jobs 
that have arrived to the system up to time $t$. 
In the \emph{online machine-minimizing job scheduling} problem, we wish to find a feasible schedule for a given set of jobs $\MC{J}$ such that the number of machines required up to time $t$ is small with respect to $OPT(\MC{J}(t))$ for any $t\ge 0$.

\begin{definition}[Competitive Factor of an Online Algorithm~\cite{Borodin:1998:OCC:290169}]
An online algorithm $\Gamma$ is said to be $c$-competitive for an optimization problem $\Pi$ if for any instance $\MC{I}$ of $\Pi$, we have $\Gamma(\MC{I}) \le c\cdot Opt(\MC{I})+x$, where $\Gamma(\MC{I})$ and $OPT(\MC{I})$ are the values computed by $\Gamma$ and the optimal solution for $\MC{I}$, respectively, and $x$ is a constant. The \emph{asymptotic} competitive factor of $\Gamma$ is defined to be $$\limsup_{n\rightarrow\infty}\BIGBP{\frac{\Gamma(\MC{I})}{n} \: : \: \text{$\MC{I}$ is an instance of $\Pi$ such that $Opt(\MC{I}) = n$.}}.$$
\end{definition}

\subsection{Other Notations and Special Job Sets.}

Let $\MC{J}$ be a set of jobs. For any $t \ge 0$, we use $\hat{\rho}\BIGP{\MC{J},t}$ to denote the maximum density among those time intervals containing $t$ with respect to $\MC{J}$, i.e., 
$$\hat{\rho}\BIGP{\MC{J},t} = \max_{0\le \ell\le t<r} \rho\BIGP{\MC{J},\ell,r}.$$
Furthermore, we use $\mathrm{def}\BIGP{\MC{J},t} = \BIGCP{\ell\BIGP{\MC{J},t},r\BIGP{\MC{J},t}}$ to denote the specific time interval that achieves the maximum density in the defining domain of $\hat{\rho}\BIGP{\MC{J},t}$.
If there are more than one such an interval, $\mathrm{def}\BIGP{\MC{J},t}$ is defined to be the one with the smallest left-end.
We call $\mathrm{def}\BIGP{\MC{J},t}$ the \emph{defining interval} of $\hat{\rho}\BIGP{\MC{J},t}$.
%
%
In addition, we define $\hat\varrho\BIGP{\MC{J}} = \max_{t\ge 0}\hat\rho\BIGP{\MC{J},t}$ to denote the maximum density for a job set $\MC{J}$.
Notice that, by Lemma~\ref{lemma-opt-offline-density}, $\hat\varrho\BIGP{\MC{J}}$ is an alternative definition of $OPT(\MC{J})$.

For ease of presentation, throughout this paper, we use a pair $\BIGP{d,\sigma}$ to denote a special problem instance for which the jobs have a universal deadline $d$, where $\sigma = \BIGP{\sigma(0),\sigma(1),\sigma(2),\ldots,\sigma(d-1)}$ is a sequence of length $d$ such that $\sigma(t)$ is the number of jobs arriving at time $t$.
%
%
Note that, under this setting, every defining interval has a right-end $d$.

\subsection{Technical Helper Lemmas}

Below, we list technical lemmas that are used in the remaining content to help prove our main results. We recommend the readers to move on to the next section and come back for further reference when referred.

\begin{lemma}
\label{prop-density-advancing}
For any $p,q,r,s \in \MBB{R}^+\cup\BIGBP{0}$, $q,s > 0$, if 
$$\frac{p}{q} \ge \frac{r}{s}, \quad \text{then} \quad \frac{p+r}{q+s} \ge \frac{r}{s}.$$
\end{lemma}

\begin{proof}
We have $$\frac{p+r}{q+s} - \frac{r}{s} = \frac{s(p+r) - r(q+s)}{s(q+s)} = \frac{ps-qr}{s(q+s)} \ge 0,$$
where the last inequality follows from the assumption $\frac{p}{q}\ge \frac{r}{s}$.
\qed
\end{proof}

\medskip

\begin{lemma}
\label{lemma-def-le-cut}
For any job set $\MC{J} = \BIGP{d,\sigma}$, any $t \ge 0$, and any $t^*$ with $\ell\BIGP{\MC{J},t} < t^* \le t$, we have $$\frac{\sum_{\ell\BIGP{\MC{J},t}\le i<t^*}\sigma(i)}{t^*-\ell\BIGP{\MC{J},t}} \ge \frac{\sum_{t^*\le i<d}\sigma(i)}{d-t^*}.$$
\end{lemma}

\begin{proof}
For simplicity and ease of presentation, let 
\begin{align*}
\begin{cases}
p = \sum_{\ell\BIGP{\MC{J},t}\le i<t^*}\sigma(i), \\
q = t^*-\ell\BIGP{\MC{J},t},
\end{cases}
&
\begin{cases}
r = \sum_{t^*\le i<d}\sigma(i), \quad \text{and} \\
s = d-t^*.
\end{cases}
\end{align*}
In this lemma, we want to prove that $$\frac{p}{q} \ge \frac{r}{s}.$$
Assume for contradiction that $$\frac{p}{q} < \frac{r}{s},$$
which implies that $ps < qr$. Then we have 
\begin{align*}
\rho\BIGP{\MC{J},t^*,d} - \rho\BIGP{\MC{J},\ell(\MC{J},t),d} \: & = \: \frac{\sum_{t^*\le i<d}\sigma(i)}{d-t^*} - \frac{\sum_{\ell\BIGP{\MC{J},t}\le i<d}\sigma(i)}{d-\ell\BIGP{\MC{J},t}} \\
& = \: \frac{r}{s} - \frac{p+r}{q+s} \\
& = \: \frac{1}{s(q+s)}\cdot(qr+rs - ps-rs) \\
& = \: \frac{1}{s(q+s)}\cdot(qr-ps),
\end{align*}
which is strictly greater than zero by our assumption. This means that $\rho\BIGP{\MC{J},t^*,d} > \rho\BIGP{\MC{J},\ell(\MC{J},t),d}$, a contradiction to the fact that $\BIGP{\ell(\MC{J},t),d}$ is the defining interval of time $t$ for which the workload density is maximized. Therefore we have $\frac{p}{q} \ge \frac{r}{s}$.
\qed
\end{proof}

\medskip

\begin{figure*}[t]
\centering
{\includegraphics[scale=1]{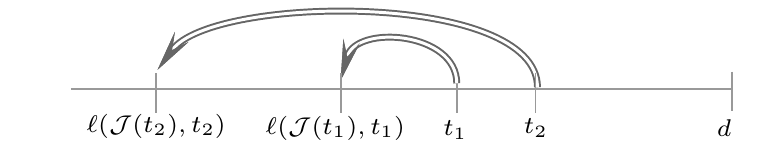}}
\caption{An illustration on the relative position of $t_1$ and $t_2$.}
\label{fig-def-le-nesting}
\end{figure*}

\begin{lemma}
\label{lemma-def-le-nesting}
For any job set $\MC{J} = \BIGP{d,\sigma}$, and any $t_1,t_2$ with $0\le t_1 < t_2 < d$, we have
$\ell\BIGP{\MC{J}(t_1),t_1} \le \ell\BIGP{\MC{J}(t_2),t_2}$.
\end{lemma}

\begin{proof}
Assume for contradiction that $\ell\BIGP{\MC{J}(t_1),t_1} > \ell\BIGP{\MC{J}(t_2),t_2}$. To get a more clear idea on the following arguments, also refer to Fig.~\ref{fig-def-le-nesting} for an illustration on the relative positions. By Lemma~\ref{lemma-def-le-cut}, with $t^*$ replaced by $\ell(\MC{J}(t_1),t_1)$, $\MC{J}$ replaced by $\MC{J}(t_2)$, and $t$ replaced by $t_2$, we get $$\frac{\sum_{\ell\BIGP{\MC{J}(t_2),t_2}\le i<\ell(\MC{J}(t_1),t_1)}\sigma(i)}{\ell(\MC{J}(t_1),t_1)-\ell\BIGP{\MC{J}(t_2),t_2}} \ge \frac{\sum_{\ell(\MC{J}(t_1),t_1)\le i\le t_2}\sigma(i)}{d-\ell(\MC{J}(t_1),t_1)},$$
whereas the latter term is at least 
$$\frac{\sum_{\ell(\MC{J}(t_1),t_1)\le i\le t_1}\sigma(i)}{d-\ell(\MC{J}(t_1),t_1)} = \hat\rho\BIGP{\MC{J}(t_1),t_1}.$$
By the above two inequalities and Lemma~\ref{prop-density-advancing}, we get
\begin{align*}
\frac{\sum_{\ell\BIGP{\MC{J}(t_2),t_2}\le i<\ell(\MC{J}(t_1),t_1)}\sigma(i)+\sum_{\ell(\MC{J}(t_1),t_1)\le i\le t_1}\sigma(i)}{\ell(\MC{J}(t_1),t_1)-\ell\BIGP{\MC{J}(t_2),t_2} + d-\ell(\MC{J}(t_1),t_1)} \ge \frac{\sum_{\ell(\MC{J}(t_1),t_1)\le i\le t_1}\sigma(i)}{d-\ell(\MC{J}(t_1),t_1)}.
\end{align*}
The left-hand side is exactly $\rho\BIGP{\MC{J}(t_1), \ell\BIGP{\MC{J}(t_2),t_2}, d}$.
Therefore, we get $$\rho\BIGP{\MC{J}(t_1), \ell\BIGP{\MC{J}(t_2),t_2}, d} \ge \hat\rho\BIGP{\MC{J}(t_1),t_1},$$ a contradiction to the fact that $\BIGP{\ell(\MC{J}(t_1),t_1),d}$ is the defining interval of $t_1$ with respect to $\MC{J}(t_1)$. Hence we must have $\ell\BIGP{\MC{J}(t_1),t_1} \le \ell\BIGP{\MC{J}(t_2),t_2}$.
\qed
\end{proof}



\section{Problem Complexity}

This section presents a lower bound of the competitive factor for the studied problem.  We consider a special case for which the jobs have a universal deadline, which will later serve as a basis to our main algorithm. In \S\ref{subsec-counter-example}, we show why the online algorithm, provided in~\cite{Shi:2008:OBP:1328331.1328337} for this special case and claimed to be optimally $2$-competitive, fails to produce feasible schedules. 
Built upon the idea behind the counter-example, we then prove a lower bound of $2.09$ for the competitive factor of any online algorithm in \S\ref{subsec-lower-bound}.
We begin with the following lemma, which draws up the curtain on the difficulty of this problem led by unknown job arrivals.

\begin{lemma}[\emph{$2$-competitivity lower bound~\bf{\cite{Shi:2008:OBP:1328331.1328337}}}]
Any online algorithm for the machine-minimizing job scheduling problem with unit jobs and a universal deadline has a competitive factor of at least $2$.
\end{lemma}

\begin{proof}
Below we sketch the proof provided in~\cite{Shi:2008:OBP:1328331.1328337}. Let $d \ge 1$ be an arbitrary integer. Consider the job set $\MC{J}^*_2 = \BIGP{d, \sigma^*_2}$, where $\sigma^*_2 = \BIGP{d,d,d,\ldots,d}$ is a sequence containing $d$ elements of value $d$. 
Notice that, we have $\mathrm{OPT}\BIGP{\MC{J}^*_2(t)} = t+1$ for all $0\le t<d$. Therefore any online algorithm with competitive factor $R$ uses at most $R\cdot\BIGP{t+1}$ machines at time $t$. Taking the summation over $0\le t<d$, we know that any online algorithm with competitive factor $R$ can schedule at most $\sum_{0\le t<d}R\cdot\BIGP{t+1} = \frac{1}{2}d(d+1)R$ jobs. Since there are $d^2$ jobs in total, we conclude that $R\ge 2$ when $d$ goes to infinity.
\qed
\end{proof}

\subsection{Why the Known Algorithm Fails to Produce Feasible Schedules}
\label{subsec-counter-example}

This section presents a counter-example for the online algorithm provided in~\cite{Shi:2008:OBP:1328331.1328337}, which we will also refer to as the \emph{packing-via-density} algorithm in the following. Given a problem instance $\MC{J} = \BIGP{d,\sigma}$, \emph{packing-via-density} works as follows.

\smallskip

At any time $t$, $t \ge 0$, the algorithm computes the maximum density with respect to the current job set $\MC{J}(t)$. More precisely, it computes $\hat{\rho}\BIGP{\MC{J}(t),t}$. Then the algorithm assigns $2\cdot\CEIL{\hat{\rho}\BIGP{\MC{J}(t),t}}$ jobs for execution.

\smallskip

Intuitively, in the computation of density, the workload of each job is equally distributed to the time interval from its arrival till its deadline, or, possibly to a larger super time interval containing it if this leads to a higher density. 
The algorithm uses another factor of $\hat{\rho}\BIGP{\MC{J}(t),t}$ in order to cover the unknown future job arrivals, which seems to be a good direction for getting a feasible scheduling.

%

However, in~\cite{Shi:2008:OBP:1328331.1328337}, the authors claimed that the job set $\MC{J}^*_2$ represents one of the worst case scenarios, followed by sketching the feasibility of \emph{packing-via-density} on $\MC{J}^*_2$.
Although from intuition this looks promising, and, by suitably defining the potential function, one can indeed prove the feasibility of \emph{packing-via-density} for $\MC{J}^*_2$ and other similar job sets, the job set $\MC{J}^*_2$ they considered is in fact not a worst scenario. The main reason is that, for each $t$ with $0\le t < d$, the left-end of the defining interval for $t$ is always zero. That is, we have $\mathrm{def}\BIGP{\MC{J}^*_2(t),t} = \BIGCP{0,d}$ for all $0\le t<d$. 
This implicitly takes all job arrivals into consideration when computing the densities.
When the sequence is more complicated and the defining intervals change over time, using $2\cdot\CEIL{\hat{\rho}\BIGP{\MC{J}(t),t}}$ machines is no longer able to cover the \emph{unpaid debt} created before $\ell\BIGP{\MC{J}(t),t}$, i.e., the jobs which are not yet finished but no longer contributing to the computation of $\hat\varrho\BIGP{\MC{J}(t)}$.
This is illustrated by the following example.

Consider the job set $\MC{J}^* = \BIGP{32, \sigma^*}$, where $\sigma^*$
is defined as
$$\sigma^* = \BIGP{\underbrace{75,75,\ldots,75}_{0 \sim 15},\: 1200,\: 0,0,0,\: \underbrace{300,300,\ldots,300}_{20 \sim 31}}$$
Notice that, for $0\le t\le 19$, we have $\mathrm{def}\BIGP{\MC{J}^*(t),t} = \BIGCP{0,32}$, and for $20\le t\le 31$, we have $\mathrm{def}\BIGP{\MC{J}^*(t),t} = \BIGCP{16,32}$.
%
%
%

\begin{theorem}
Using algorithm \emph{packing-via-density} on the job set $\MC{J}^*$ results in deadline misses of $10$ jobs.
\end{theorem}

\begin{proof}
We prove by calculating the density for each moment explicitly.

\begin{itemize}
\item
For $0\le t \le 15$, we have $$\rho\BIGP{\MC{J}^*(t), t, 32} = \frac{75}{32-t}, $$
which is strictly smaller than $\sigma^*(\ell) = 75$, for all $0\le \ell \le t$.
Therefore, by Lemma~\ref{prop-density-advancing}, decreasing the value of $\ell$ in $\rho\BIGP{\MC{J}^*(t), \ell, 32}$ always increases this density, and $\rho\BIGP{\MC{J}^*(t), \ell, 32}$ is maximized at $\ell = 0$ with the value $$\rho\BIGP{\MC{J}^*(t), 0, 32} = \frac{75}{32}\BIGP{t+1}.$$

For $0 \le t \le 15$, the values of $\CEIL{\hat\rho\BIGP{\MC{J}^*(t),t}}$ are
$$3, \: 5, \: 8, \: 10, \: 12, \: 15, \: 17, \: 19, \: 22, \: 24, \: 26, \: 29, \: 31, \: 33, \: 36, \: 38,$$
respectively, and $\sum_{0\le t\le 15}2\cdot\CEIL{\hat\rho\BIGP{\MC{J}^*(t),t}} = 656$.

\medskip

\item
For $t = 16$, consider each $0\le \ell \le 15$, we have
$$\rho\BIGP{\MC{J}^*(t), 16, 32} = \frac{1200}{16} = 75 = \sigma^*(\ell).$$
Therefore the function $\rho\BIGP{\MC{J}^*(t), \ell, 32}$ is again maximized at $\ell = 0$ by Lemma~\ref{prop-density-advancing}, and we have $2\cdot\CEIL{\hat\rho\BIGP{\MC{J}^*(16),16}} = 150$.

\medskip

\item
For $17 \le t \le 19$, consider each $17 \le \ell \le 19$, we have $\rho\BIGP{\MC{J}^*(t),\ell,32} = 0$.
For $0\le \ell \le 16$, the situation is identical to the case when $t = 16$. Therefore we have
$\sum_{17 \le t \le 19}2\cdot\CEIL{\hat\rho\BIGP{\MC{J}^*(t),t}} = 3\cdot 150 = 450$.

\medskip

\item
For $20 \le t \le 31$, we have $$\rho\BIGP{\MC{J}^*(t), t, 32} = \frac{300}{32-t}, $$
which is smaller or equal to $\sigma^*(\ell) = 300$, for all $20 \le \ell \le t$. Since we also have $\frac{1}{4}\cdot\sum_{16\le \ell \le 19}\sigma^*(\ell) = 300$, 
by Lemma~\ref{prop-density-advancing}, we get
$$\rho\BIGP{\MC{J}^*(t), \ell, 32} \le \rho\BIGP{\MC{J}^*(t), 16, 31} = \frac{1}{16}\cdot\BIGP{1200 + 300\cdot(t-19)}.$$
Since we have $$\frac{1}{16}\cdot\BIGP{1200 + 300\cdot(t-19)} > 75 = \sigma^*(\ell), \quad \text{for all $0\le \ell \le 15$,}$$
by Lemma~\ref{lemma-def-le-cut}, we know that $\ell\BIGP{\MC{J}^*(t),t} > 15$. Therefore we have
$$\hat\rho\BIGP{\MC{J}^*(t),t} = \rho\BIGP{\MC{J}^*(t), 16, 31} = \frac{1}{16}\cdot\BIGP{1200 + 300\cdot(t-19)}.$$
For $20\le t\le 31$, the values of $\CEIL{\hat\rho\BIGP{\MC{J}^*(t),t}}$ are
$$94, \: 113, \: 132, \: 150, \: 169, \: 188, \: 207, \: 225, \: 244, \: 263, \: 282, \: 300,$$
respectively, and $\sum_{20\le t\le 31}2\cdot\CEIL{\hat\rho\BIGP{\MC{J}^*(t),t}} = 4734$.
\end{itemize}

\bigskip

By the above discussion, we conclude that $$\sum_{0\le t\le 31}2\cdot\CEIL{\hat\rho\BIGP{\MC{J}^*(t),t}} = 656 + 150 + 450 + 4734 = 5990.$$
Since the total number of jobs is $75\times 16 + 1200 + 300 \times 12 = 6000$, we have exactly $10$ jobs which fail to finish their execution by their deadlines.
\qed
\end{proof}

\subsection{Lower Bound on the Competitive Factor}
\label{subsec-lower-bound}

In fact, by further generalizing the construction of $\MC{J}^*$, we can design an \emph{online adversary} that proves a lower bound strictly greater than $2$ for the competitive factor of any online algorithm.
In the job set $\MC{J}^*$, the debt is created by making a one-time change of the defining interval at time $20$. Below, we construct an example whose defining intervals can alter for arbitrarily many times, thereby creating sufficiently large debts. Then, we present our online adversary.

\begin{figure*}[t]
{\includegraphics[scale=1.2]{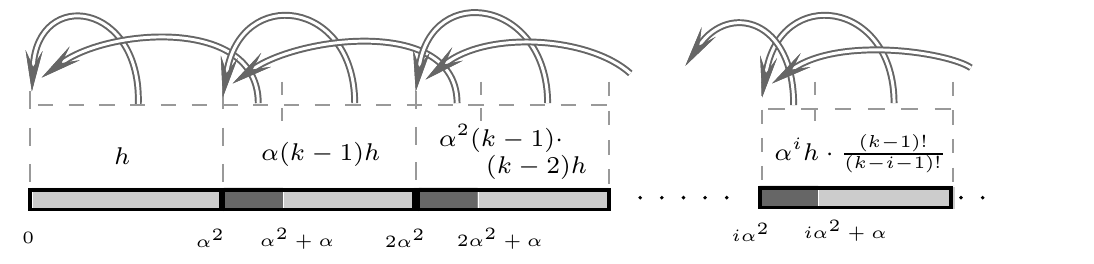}}
\caption{An illustration for the arrival sequence $\sigma^*_{k,\alpha}$ of $\MC{J}^*_{k,\alpha}$ and the changes of the left-ends of defining intervals over time.}
\label{fig-lower-bound-seq}
\end{figure*}

Let $h,k,\alpha \in \MBB{N}$ be three constants to be decided later. We define the job set $\MC{J}^*_{k,\alpha} = \BIGP{k\alpha^2, \sigma^*_{k,\alpha}}$ as follows.
For each $0\le i < k$ and each $0\le j < \alpha^2$,
$$\sigma^*_{k,\alpha}\BIGP{i\alpha^2+j} = \begin{cases}
h, & \text{if} \enskip i = 0, \\
\alpha\BIGP{k-i}\cdot \sigma^*_{k,\alpha}\BIGP{(i-1)\alpha^2+j}, & \text{otherwise.}
\end{cases}$$
%
Also refer to Fig.~\ref{fig-lower-bound-seq} for an illustration of the sequence.
The following lemma shows the changes of the defining intervals over time.

\begin{lemma}
\label{lemma-def-le}
For each time
$t = i\alpha^2 + j$, where $0\le i<k$, $0\le j<\alpha^2$, we have
$$\ell\BIGP{\MC{J}^*_{k,\alpha}(t), t} = \begin{cases}
0, & \text{for} \enskip i = 0, \\ 
\BIGP{i-1}\alpha^2, \quad & \text{for} \enskip i > 0 \enskip \text{and} \enskip 0\le j<\alpha, \\
i\alpha^2, & \text{for} \enskip i > 0 \enskip \text{and} \enskip \alpha\le j < \alpha^2.
\end{cases}$$
\end{lemma}

\begin{proof}
Depending on the values of $i$ and $j$, we consider the following cases to obtain lower-bounds and upper-bounds on the left-ends of the defining intervals.
\begin{itemize}
\item[Case (1):]
Consider each $i$ and $j$ with $0 \le i<k$ and $0\le j< \alpha^2$, and any $j^\prime$ with $0 \le j^\prime \le j$.
The density of the time interval $\BIGCP{i\alpha^2+j^\prime, k\alpha^2}$ with respect to job set $\MC{J}^*_{k,\alpha}(i\alpha^2+j)$ is 
\begin{equation}
\rho\BIGP{\MC{J}^*_{k,\alpha}(i\alpha^2+j), i\alpha^2+j^\prime, k\alpha^2} = \frac{j-j^\prime+1}{(k-i)\alpha^2-j^\prime}\cdot\sigma^*_{k,\alpha}(i\alpha^2+j).
\label{eq-lower-bound-le-1}
\end{equation}
Since $i < k$ and $j^\prime \le j < \alpha^2$, we know that 
$$\frac{(j+1)-j^\prime}{(k-i)\alpha^2-j^\prime} \le \frac{\alpha^2-j^\prime}{\alpha^2-j^\prime} = 1.$$
Combining with Eq.~(\ref{eq-lower-bound-le-1}), we get
$\rho\BIGP{\MC{J}^*_{k,\alpha}(i\alpha^2+j), i\alpha^2+j^\prime, k\alpha^2} \le \sigma^*_{k,\alpha}(i\alpha^2+j) = \sigma^*_{k,\alpha}(i\alpha^2+j^\prime)$. This holds for all $j^\prime$ with $0\le j^\prime\le j$.
Therefore, by Prop.~\ref{prop-density-advancing}, we know that 
$$\rho\BIGP{\MC{J}^*_{k,\alpha}(i\alpha^2+j), i\alpha^2, k\alpha^2} \ge \rho\BIGP{\MC{J}^*_{k,\alpha}(i\alpha^2+j), i\alpha^2+j^\prime, k\alpha^2}.$$ 
Hence, we get
\begin{equation}
\ell\BIGP{\MC{J}^*_{k,\alpha}(i\alpha^2+j),i\alpha^2+j} \le i\alpha^2.
\label{ieq-le-1}
\end{equation}
This shows that $\ell\BIGP{\MC{J}^*_{k,\alpha}(j),j} = 0$ for each $0\le j < \alpha^2$.

\smallskip

\item[Case (2):]
Next, consider each $i$ and $j$ with $0 < i < k$ and $\alpha \le j < \alpha^2$.
We have
\begin{align*}
\rho\BIGP{\MC{J}^*_{k,\alpha}(i\alpha^2+j), i\alpha^2, k\alpha^2} \: & \ge \: \rho\BIGP{\MC{J}^*_{k,\alpha}(i\alpha^2+\alpha), i\alpha^2, k\alpha^2} \\
& = \: \frac{\alpha+1}{(k-i)\alpha^2}\cdot\sigma^*_{k,\alpha}(i\alpha^2) \\
& > \: \frac{1}{(k-i)\alpha}\cdot\sigma^*_{k,\alpha}(i\alpha^2) \: = \: \sigma^*_{k,\alpha}(i\alpha^2-1).
\end{align*}
Since we have $\sigma^*_{k,\alpha}(i^\prime\alpha^2+j^\prime) \le \sigma^*_{k,\alpha}(i\alpha^2-1)$ for all $0\le i^\prime < i$ and $0\le j^\prime < \alpha^2$ by our design, together with Lemma~\ref{lemma-def-le-cut}, we get $\ell\BIGP{\MC{J}^*_{k,\alpha}(i\alpha^2+j),i\alpha^2+j} > i\alpha^2-1$.
Combining with Ineq.~(\ref{ieq-le-1}) from case (1), we get
\begin{equation}
\ell\BIGP{\MC{J}^*_{k,\alpha}(i\alpha^2+j),i\alpha^2+j} = i\alpha^2.
\label{eq-le-2}
\end{equation}

\smallskip

\item[Case (3):]
Now, consider each $i$ and $j$ with $0 < i < k$ and $0\le j<\alpha$, from Eq.~(\ref{eq-lower-bound-le-1}) we know that
\begin{align*}
\rho\BIGP{\MC{J}^*_{k,\alpha}(i\alpha^2+j), i\alpha^2,k\alpha^2} \: & = \: \frac{j+1}{(k-i)\alpha^2}\cdot\sigma^*_{k,\alpha}(i\alpha^2+j) \\
& \le \: \frac{1}{(k-i)\alpha}\cdot\sigma^*_{k,\alpha}(i\alpha^2+j) \: = \: \sigma^*_{k,\alpha}(i\alpha^2-1).
\end{align*}
This shows that, for this case, $\ell\BIGP{\MC{J}^*_{k,\alpha}(i\alpha^2+j), i\alpha^2+j} \le (i-1)\alpha^2$. By Eq.~(\ref{eq-le-2}) in case (2) and Lemma~\ref{lemma-def-le-nesting}, we get $\ell\BIGP{\MC{J}^*_{k,\alpha}(i\alpha^2+j), i\alpha^2+j} = (i-1)\alpha^2$.
\end{itemize}
This proves the lemma. 
\qed
\end{proof}

Below we present our online adversary, which we denote by $\MC{A}^*(c)$, where $c>0$ is a constant. Let $\Gamma$ be an arbitrary feasible online scheduling algorithm for this problem. The adversary works as follows.
At each time $t$ with $0\le t < k\alpha^2$, $\MC{A}^*(c)$ releases $\sigma^*_{k,\alpha}(t)$ jobs with deadline $k\alpha^2$ for algorithm $\Gamma$ and observes the behavior of $\Gamma$. If $\Gamma$ uses more than $c\cdot\hat\varrho\BIGP{\MC{J}^*_{k,\alpha}(t)}$ machines, then $\MC{A}^*(c)$ terminates immediately. Otherwise, $\MC{A}^*(c)$ proceeds to time $t+1$ and repeats the same procedure. This process continues till time $k\alpha^2$.
%

%

\begin{theorem}
Any feasible online algorithm for the machine-minimizing job scheduling problem has a competitive factor at least $2.09$, even for the case when the jobs have a universal deadline. 
\end{theorem}

\begin{proof}
Consider the adversary $\MC{A}^*(c)$ and any feasible online algorithm $\Gamma$ for this problem. Let $t$ be the time for which $\MC{A}^*(c)$ terminates its execution.
If $t < k\alpha^2$, then, by Lemma~\ref{lemma-opt-offline-density}, the competitive factor of $\Gamma$ with respect to the input job set, $\MC{J}^*_{k,\alpha}(t)$, is strictly greater than $c$.

On the other hand, if $t = k\alpha^2$, then the number of jobs $\Gamma$ can schedule is at most $\sum_{0\le i<k}\sum_{0\le j<\alpha^2}c\cdot\hat\varrho\BIGP{\MC{J}^*_{k,\alpha}(i\alpha^2+j)}$.
According to Lemma~\ref{lemma-def-le}, we know that:
\begin{itemize}
\item
For $i = 0$ and $0\le j < \alpha^2$,
$$\hat\varrho\BIGP{\MC{J}^*_{k,\alpha}(i\alpha^2+j)} = \rho\BIGP{\MC{J}^*_{k,\alpha}(j), 0,k\alpha^2} = \frac{j+1}{k\alpha^2}\cdot h.$$

\item
For $0 < i < k$ and $0\le j < \alpha$,
\begin{align*}
\hat\varrho\BIGP{\MC{J}^*_{k,\alpha}(i\alpha^2+j)} & = \rho\BIGP{\MC{J}^*_{k,\alpha}(i\alpha^2+j), (i-1)\alpha^2,k\alpha^2} \\
& = \frac{1}{(k-i+1)\alpha^2}\cdot\BIGP{\alpha^2 + (j+1)\cdot\alpha(k-i)}\cdot\alpha^{i-1}h\cdot\frac{(k-1)!}{(k-i)!}.
\end{align*}

\item
For $0 < i < k$ and $\alpha \le j < \alpha^2$,
$$\hat\varrho\BIGP{\MC{J}^*_{k,\alpha}(i\alpha^2+j)} = \rho\BIGP{\MC{J}^*_{k,\alpha}(i\alpha^2+j), i\alpha^2,k\alpha^2} = \frac{j+1}{(k-i)\alpha^2}\cdot \alpha^ih\cdot\frac{(k-1)!}{(k-i-1)!}.$$
\end{itemize}

By choosing $k$ to be $6$ and $\alpha$ to be $5$, a direct computation shows that 
$$\sum_{0\le i<k}\sum_{0\le j<\alpha^2}c\cdot\hat\varrho\BIGP{\MC{J}^*_{k,\alpha}(i\alpha^2+j)} \approx c\cdot 5.47695 \times 10^6h,$$
while the total number of jobs is $$\sum_{0\le i< k}\alpha^2\cdot\alpha^ih\cdot\frac{(k-1)!}{(k-i-1)!} \approx 1.14506 \times 10^7h.$$
By choosing $c = 2.09$, from the direct computation we know that $$\sum_{0\le i<k}\sum_{0\le j<\alpha^2}c\cdot\hat\varrho\BIGP{\MC{J}^*_{k,\alpha}(i\alpha^2+j)} < \sum_{0\le i< k}\alpha^2\cdot\alpha^ih\cdot\frac{(k-1)!}{(k-i-1)!},$$
which is a contradiction to the assumption that $\Gamma$ is a feasible online algorithm. Therefore, we must have $t < k\alpha^2$, and $\Gamma$ has competitive factor at least $c = 2.09$. Since this argument holds for any $h > 0$, this implies a lower bound of $2.09$ on both the competitive factor and the asymptotic competitive factor of any online algorithm.

We remark that, although it may seem confusing that the densities could be fractional numbers while we need integral number of machines,  
this can be avoided by setting $h$ to be $k!\alpha^2\cdot h^\prime$, for some positive interger $h^\prime$. This will make integral densities and we still get the same conclusion.
\qed
\end{proof}

%


\section{$5.2$-Competitive Packing-via-Density}

As indicated in Lemma~\ref{lemma-opt-offline-density}, to come up with a good scheduling algorithm for online machine-minimizing job scheduling with unit jobs, it suffices to compute a good approximation of the offline density for the entire job set, as this corresponds directly to the number of machines required by any optimal schedule.
From the proofs for the problem complexity in Section~\ref{subsec-lower-bound}, for any job set $\MC{J}$ and $t\ge 0$, the gap between $\hat{\varrho}\BIGP{\MC{J}}$, which is the maximum offline density for the entire job set, and $\hat{\varrho}\BIGP{\MC{J}(t)}$, which is the maximum density the online algorithm for the jobs arrived before and at time $t$, can be arbitrarily large. For instance, in the simple job set $\MC{J}^*_2$, we have $\hat\varrho\BIGP{\MC{J}^*_2} = d$ while $\hat\varrho\BIGP{\MC{J}^*_2(t)} = t+1$ for all $0\le t<d$.

In~\cite{Shi:2008:OBP:1328331.1328337}, the authors proved that, simply using $\CEIL{\hat\varrho\BIGP{\MC{J}(t)}}$ to approximate the offline density as suggested in the classical mainstream \emph{any fit} algorithms, such as \emph{best fit}, \emph{first fit}, etc., can results in the deadline misses of $\Theta\BIGP{\log n}$ jobs even for the job set $\MC{J}^*_2$, meaning that we will have to use $\Theta\BIGP{\log n}$ machines in the very last moment in order to prevent deadline misses if we apply these classical packing algorithms. 
Therefore, additional space sparing at each moment is necessary for obtaining a better approximation guarantee on the offline density in later times.
%
%
One natural question to ask is: 

\noindent
\begin{quote} 
\emph{Is there a constant $c$ such that $c\cdot\hat\varrho\BIGP{\MC{J}(t)}$ is an approximation of $\hat\varrho\BIGP{\MC{J}}$ for all $t\ge 0$?}
\end{quote}

In this section, we give a positive answer to the above question in a slightly more general way.
We show that, with a properly chosen constant $c$, using $\CEIL{c\cdot\hat\varrho\BIGP{\MC{J}(t)}}$ machines at all times gives a feasible scheduling which is also $c$-competitive for the machine-minimizing job scheduling with unit job execution time and \emph{arbitrary job deadlines}.

\paragraph{The \emph{packing-via-density($c$)} algorithm.}

At time $t$, $t \ge 0$, the algorithm computes the 
maximum density it has seen so far, i.e., $\hat\varrho\BIGP{\MC{J}(t)}$.
Then the algorithm assigns $\CEIL{c\cdot\hat\varrho\BIGP{\MC{J}(t)}}$ jobs with earliest deadlines form the ready queue for execution.

\bigskip

Since $\hat\varrho\BIGP{\MC{J}(t)} \le \hat\varrho\BIGP{\MC{J}}$ for all $t \ge 0$, by Lemma~\ref{lemma-opt-offline-density}, we have $\hat\varrho\BIGP{\MC{J}(t)} \le \mathrm{OPT}\BIGP{\MC{J}}$ for all $t\ge 0$. Hence, we know that the schedule produced by \emph{packing-via-density(c)} is $c$-competitive as long as it is feasible.
%
In the following, we show that, for a properly chosen constant $c$, \emph{packing-via-density(c)} always produces a feasible schedule for any upcoming job set. To this end, for any job set, we present two reductions to obtain a sequence whose structure is relatively simple in terms of the altering of defining intervals,
followed by providing a direct analysis on the potential function of that sequence.

\paragraph{Feasibility of \emph{packing-via-density(c)}.}

For any job set $\MC{J}$, any $t_1,t_2$ with $0\le t_1<t_2$, and any $c>0$, consider the potential function $\Phi_c\BIGP{\MC{J},t_1,t_2}$ defined as
$$\Phi_c\BIGP{\MC{J},t_1,t_2} = 
\BIGP{\sum_{t_1\le t<t_2}c\cdot\hat\varrho\BIGP{\MC{J}(t)}} - \MC{J}\BIGP{t_1,t_2}.$$
Literally, in this potential function we consider the sum of maximum densities over each moment between time $t_1$ and time $t_2$, subtracted by the total amount of workload which arrives and has to be done within the time interval $\BIGCP{t_1,t_2}$. 
%
Since \emph{packing-via-density(c)} assigns $\CEIL{c\cdot\hat\varrho\BIGP{\MC{J}(t)}}$ jobs for execution for any moment $t$,
by Lemma~\ref{lemma-feasibility-characterization}, we have the feasibility of this algorithm if and only if $\Phi_c\BIGP{\MC{J},t_1,t_2} \ge 0$ for all $0\le t_1<t_2$.

Below, we present our first reduction and show that, it suffices to prove the non-negativity of this potential function for any job set with a universal deadline.
For any $d \ge 0$, consider the job set $\MC{J}_d$ defined as follows.
For each $\BIGP{i,j} \in \MC{J}$ such that $j\le d$, we create a job $\BIGP{i,d}$ and put it into $\MC{J}_d$.

%

\begin{lemma}[\emph{Reduction to the case of equal deadlines}]
\label{lemma-reduction-equal-deadline}
For any $c\ge 0$, $d\ge 0$, and $0\le t<d$, we have $\Phi_c\BIGP{\MC{J}_d,t,d} \le \Phi_c\BIGP{\MC{J},t,d}$.
\end{lemma}

\begin{proof}
Consider the job set $\MC{J}(t^*)$ for any moment $t^*$ with $t \le t^* < d$. For any time interval $\BIGCP{t_1,t_2}$ such that $0\le t_1<t_2$, by the definition of $\MC{J}_d$, we know that: (1) If $t_2 < d$, then $\BIGP{\MC{J}_d(t^*)}(t_1,t_2) = 0$. (2) If $t_2 = d$, then $\BIGP{\MC{J}_d(t^*)}(t_1,t_2) = \BIGP{\MC{J}(t^*)}(t_1,t_2)$. (3) If $t_2 > d$, then 
$$\BIGP{\MC{J}_d(t^*)}(t_1,t_2) = \BIGP{\MC{J}(t^*)}(t_1,d) \le \BIGP{\MC{J}(t^*)}(t_1,t_2).$$

Therefore, for all cases, we have $\BIGP{\MC{J}_d(t^*)}(t_1,t_2) \le \BIGP{\MC{J}(t^*)}(t_1,t_2)$, which in turn implies $\rho\BIGP{\MC{J}_d(t^*),t_1,t_2} \le \rho\BIGP{\MC{J}(t^*),t_1,t_2}$, and $\hat\varrho\BIGP{\MC{J}_d(t^*)} \le \hat\varrho\BIGP{\MC{J}(t^*)}$.
Moreover, we have $\MC{J}_d(t,d) = \MC{J}(t,d)$. 
This proves this lemma.
\qed
\end{proof}

As the mapping from $\MC{J}$ to $\MC{J}_d$ is well-defined for each $d \ge 0$, by Lemma~\ref{lemma-reduction-equal-deadline}, the non-negativity of the potential function with respect to \emph{any job set with a universal deadline} will in turn imply the non-negativity of that with respect to the given job set $\MC{J}$.
%
Therefore, it suffices to show that,
%
for any job set $\MC{J}_d = \BIGP{d,\sigma_d}$, we have $\Phi_c\BIGP{\MC{J}_d,0,d} \ge 0$.

Next, we make another reduction and show that, it suffices to consider job sets with non-decreasing arrival sequences:
%
(i) If $\sigma_d$ is already a non-decreasing sequence, then there is nothing to argue. 
(ii) Otherwise, let $k$, $0\le k<d-1$, be the largest integer such that $\sigma_d(k) > \sigma_d(k+1)$. Furthermore, let $m$, $k < m< d$, be the largest integer such that $\sigma_d(m)$ is under the average of $\sigma_d(k), \sigma_d(k+1), \ldots, \sigma_d(m)$, i.e.,
$$\sigma_d(m) < \frac{\sum_{k\le i\le m}\sigma_d(i)}{m-k+1}, \enskip \text{and,} \enskip \sigma_d(m+1) \ge \frac{\sum_{k\le i\le m+1}\sigma_d(i)}{m-k+2} \enskip \text{if} \enskip m < d-1.$$
%
%
Consider the job set $\MC{J}_{d,k,m} = \BIGP{d,\sigma_{d,k,m}}$
defined as follows.
For each $i$ with $0\le i < k$ or $m<i<d$, we set
$\sigma_{d,k,m}(i) = \sigma_d(i)$. Otherwise, for $i$ with $k\le i\le m$, we set 
$$\sigma_{d,k,m}(i) = \frac{\sum_{k\le i\le m}\sigma_d(i)}{m-k+1}.$$
The following lemma shows that the effect of this change on the potential of the resulting job set is non-increasing.

\begin{lemma}[\emph{Reduction to non-decreasing sequence}]
\label{lemma-reduction-non-decreasing}
$$\Phi_c\BIGP{\MC{J}_{d,k,m},0,d} \le \Phi_c\BIGP{\MC{J}_d,0,d}.$$
\end{lemma}

\begin{proof}
The amount of total workload in $\MC{J}_{d,k,m}$ and $\MC{J}_d$ are the same, i.e., $\MC{J}_{d,k,m}(0,d) = \MC{J}_d(0,d)$. 
Therefore, it suffices to argue that, for each $0\le t<d$, we have $\hat\varrho\BIGP{\MC{J}_{d,k,m}(t)} \le \hat\varrho\BIGP{\MC{J}_d(t)}$. Depending on the value of $t$, we consider the following three cases.

\begin{itemize}
\item For $0\le t < k$, we know that the two job sets, $\MC{J}_{d,k,m}(t)$ and $\MC{J}_d(t)$, are identical, since $\sigma_{d,k,m}(i) = \sigma_d(i)$, for all $0\le i\le t$. Hence, 
$\hat\varrho\BIGP{\MC{J}_{d,k,m}(t)} = \hat\varrho\BIGP{\MC{J}_d(t)}$.

\smallskip

\item
For $k\le t \le m$, we argue that, for all $t^*$ with $k\le t^*\le t$, we have $\ell\BIGP{\MC{J}_{d,k,m}(t), t^*} \le k$. Suppose for contradiction that $k < \ell\BIGP{\MC{J}_{d,k,m}(t), t^*} \le m$ for some $t^*$ with $k\le t^* \le t$. Then, by Lemma~\ref{lemma-def-le-cut}, we get
$$\frac{\sum_{\ell\BIGP{\MC{J}_{d,k,m}(t), t^*} \le i \le m}\sigma_{d,k,m}(i)}{m-\ell\BIGP{\MC{J}_{d,k,m}(t), t^*}+1} \ge \rho\BIGP{\MC{J}_{d,k,m}(t), \ell\BIGP{\MC{J}_{d,k,m}(t), t^*}, d}.$$
By the definition of $k$ and $m$, we know that 
$$\frac{\sum_{k \le i < \ell\BIGP{\MC{J}_{d,k,m}(t), t^*}}\sigma_{d,k,m}(i)}{\ell\BIGP{\MC{J}_{d,k,m}(t), t^*}-k} \ge \frac{\sum_{\ell\BIGP{\MC{J}_{d,k,m}(t), t^*} \le i \le m}\sigma_{d,k,m}(i)}{m-\ell\BIGP{\MC{J}_{d,k,m}(t), t^*}+1},$$
which implies that, taking the subsequence $$\BIGP{\sigma_{d,k,m}(k), \sigma_{d,k,m}(k+1), \ldots, \sigma_{d,k,m}(\ell\BIGP{\MC{J}_{d,k,m}(t), t^*}-1)}$$ into consideration can never make the density worse, i.e., $$\rho\BIGP{\MC{J}_{d,k,m}(t), k, d} \ge \rho\BIGP{\MC{J}_{d,k,m}(t), \ell\BIGP{\MC{J}_{d,k,m}(t), t^*}, d},$$
a contradiction to the definition of $\ell\BIGP{\MC{J}_{d,k,m}(t), t^*}$.

Hence, we have $\ell\BIGP{\MC{J}_{d,k,m}(t), t^*} \le k$ for all $t^*$ with $k\le t^*\le t$.
This implies that $\hat\varrho\BIGP{\MC{J}_{d,k,m}(t)} \le \hat\varrho\BIGP{\MC{J}_d(t)}$ since $\rho\BIGP{\MC{J}_{d,k,m}(t), t^*, d} \le \rho\BIGP{\MC{J}_d(t), t^*, d}$ for all $0 \le t^* \le t$.

\smallskip

\item
For $m < t < d$, by a similar argument to the above case, we know that, for all $t^*$ with $k\le t^*\le t$, we either have $\ell\BIGP{\MC{J}_{d,k,m}(t), t^*} \le k$, or $\ell\BIGP{\MC{J}_{d,k,m}(t), t^*} > m$. 

Therefore, $\hat\varrho\BIGP{\MC{J}_{d,k,m}(t)} = \hat\varrho\BIGP{\MC{J}_d(t)}$ since $\rho\BIGP{\MC{J}_{d,k,m}(t), t^*, d} = \rho\BIGP{\MC{J}_d(t), t^*, d}$ for all $0 \le t^* \le k$ and $m < t^* \le t$.
\end{itemize}
For all $0\le t < d$, we have $\hat\varrho\BIGP{\MC{J}_{d,k,m}(t)} = \hat\varrho\BIGP{\MC{J}_d(t)}$. This proves the lemma.
\qed
\end{proof}

Repeating the process described in (ii) above, we get a non-decreasing sequence $\sigma^\uparrow_d$ for the job set $\MC{J}_d$. By Lemma~\ref{lemma-reduction-non-decreasing}, we know that the non-negativity of the potential function with respect to $({d,\sigma^\uparrow_d})$ will in turn imply the non-negativity of that with respect to $\MC{J}_d$.
%

%

\smallskip

Now, we consider the job set $\MC{J}^\uparrow_d = ({d, \sigma^\uparrow_d})$
%
and provide a direct analysis on the value of 
$\Phi_c({\MC{J}^\uparrow_d,0,d})$. To take the impact of unknown job arrivals and the unpaid-debt excluded implicitly in the computation of $\hat\varrho({\MC{J}^\uparrow_d(t)})$, we exploit the non-decreasing property of the sequence and use a backward analysis. More precisely, starting from the tail, 
for each $0\le i \le \FLOOR{\log_2d}$, we consider 
the sequence 
$$\MC{S}_i = \BIGP{\sigma^\uparrow_d(d-2^i), \sigma^\uparrow_d(d-2^i+1), \ldots, \sigma^\uparrow_d(d-1)},$$
whose length grows exponentially as $i$ increases.
Also refer to Fig.~\ref{fig-8-competitive-seq} for an illustration.
The following lemma shows that a simple argument, which takes eight times the sum of densities over $\MC{S}_i$ to cover the total workload in $\MC{S}_{i+1}$, already asserts the feasibility of \emph{packing-via-density$(8)$}.

\begin{figure*}[t]
{\includegraphics[scale=1]{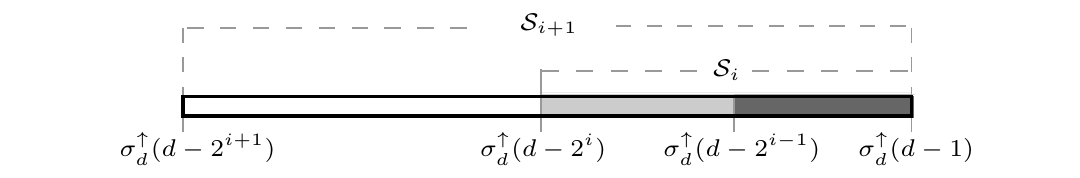}}
\caption{An illustration for the sequences $\MC{S}_i$, $0 < i \le \FLOOR{\log_2d}$. The dark-grey area denotes the sequence $\MC{S}_{i-1}$, whose densities have already been counted in the last iteration. The light-grey area denotes the portion of $\MC{S}_i$ whose densities are not counted yet. The white area denotes the portion whose workload are to be covered in this iteration.}
\label{fig-8-competitive-seq}
\end{figure*}

\begin{lemma}[\emph{Feasibility of packing-via-density$(8)$}]
$$\Phi_8\BIGP{\MC{J}^\uparrow_d,0,d} \ge 0.$$
\end{lemma}

\begin{proof}
%
For $0\le i \le \FLOOR{\log_2d}$, consider the contribution of the sequence $$\MC{S}_i = \BIGP{\sigma^\uparrow_d(d-2^i), \sigma^\uparrow_d(d-2^i+1), \ldots, \sigma^\uparrow_d(d-1)}$$
into the potential function $\Phi_8\BIGP{\MC{J}^\uparrow_d,0,d}$.
For $i > \FLOOR{\log_2d}$, $\MC{S}_i$ is defined to be a sequence of infinite zeros.
We prove this lemma by showing that, the sum of densities over $\MC{S}_i$, timed by $8$, is sufficient to cover the total workload of $\MC{S}_i$ and $\MC{S}_{i+1}$, for all $0\le i \le \FLOOR{\log_2d}$. 
Consider the following two cases.
\begin{itemize}
\item 
For $i = 0$, we have $\hat\rho\BIGP{\MC{J}^\uparrow_d, d-1} \ge \rho\BIGP{\MC{J}^\uparrow_d, d-1,d} = \sigma^\uparrow_d(d-1)$. 
By the non-decreasing property of $\sigma^\uparrow_d$, we know that $\sigma^\uparrow_d(d-2) \le \sigma^\uparrow_d(d-1)$.
Therefore, the sum of densities over $\MC{S}_0$, subtracted by the total workload in $\MC{S}_1$, is at least $$8\cdot\sigma^\uparrow_d(d-1) - \BIGP{\sigma^\uparrow_d(d-1) + \sigma^\uparrow_d(d-2)}> 0.$$

\smallskip

\item
For $0 < i \le \FLOOR{\log_2d}$, consider the element $\sigma^\uparrow_d\BIGP{d-2^i}$. From the non-decreasing property, we know that for all $d-2^i \le j < d-2^{i-1}$, $\sigma^\uparrow_d(j) \ge \sigma^\uparrow_d\BIGP{d-2^i}$. Therefore, 
$$\hat\rho\BIGP{\MC{J}^\uparrow_d(j), j} \ge \rho\BIGP{\MC{J}^\uparrow_d(j), d-2^i, d} \: \ge \: \frac{1}{2^i}\cdot\BIGP{j-(d-2^i)+1}\cdot\sigma^\uparrow_d(d-2^i).$$
Summing up $\hat\varrho\BIGP{\MC{J}^\uparrow_d(j)}$ over $d-2^i \le j < d-2^{i-1}$, we get
$$\sum_{d-2^i \le j < d-2^{i-1}}\hat\varrho\BIGP{\MC{J}^\uparrow_d(j)} \ge \frac{\sigma^\uparrow_d\BIGP{d-2^i}}{2^i}\sum_{1\le j \le 2^{i-1}}j = \frac{\sigma^\uparrow_d\BIGP{d-2^i}}{4}\BIGP{2^{i-1}+1}.$$
On the other hand, the amount of workload in $\MC{S}_{i+1} \backslash \MC{S}_i$ is at most $2^i\cdot\sigma^\uparrow_d(d-2^i)$. Hence, 
\begin{align*}
& \sum_{d-2^i \le j < d-2^{i-1}}8\cdot\hat\varrho\BIGP{\MC{J}^\uparrow_d(j)} - \sum_{d-2^{i+1}\le j< d-2^i}\sigma^\uparrow_d(j) \\
\ge \enskip & 2\BIGP{2^{i-1}+1}\sigma^\uparrow_d(d-2^i) - 2^i\cdot\sigma^\uparrow_d(d-2^i) > 0.
\end{align*}
Also refer to Fig.~\ref{fig-8-competitive-seq} for an illustration.
\end{itemize}

Finally, consider the potential function $\Phi_8\BIGP{\MC{J}^\uparrow_d,0,d}$, which is $$\sum_{0\le t<d}\BIGP{8\cdot\hat\varrho\BIGP{\MC{J}^\uparrow_d(t)} - \sigma^\uparrow_d(t)}$$ by definition. Each item in the summation is counted exactly once in our case study. Therefore we have $\Phi_8\BIGP{\MC{J}^\uparrow_d,0,d} \ge 0$.
\qed
\end{proof}

The basic argument gives a hint on the reason why the potential function can be made positive-definite by a carefully chosen constant $c$. Below we show that, a further generalized approach gives a better bound.
The idea is to further exploit the densities generated by $\MC{S}_{i+1}$ itself: when the amount of workload in $\MC{S}_{i+1}$ is relatively low, then a smaller factor from $\MC{S}_i$ suffices, and when the amount of workload in $\MC{S}_{i+1}$ becomes higher, then most of the workload can be covered by the densities $\MC{S}_{i+1}$ itself generates.
In addition, we use the exponential base of $3$ to define the subsequences instead of $2$.

\begin{lemma}
\label{lemma-5-2-potential}
$$\Phi_{5.2}\BIGP{\MC{J}^\uparrow_d,0,d} \ge 0.$$
\end{lemma}

\begin{proof}
Let $\alpha \in \MBB{N}$, $\beta, \lambda_1, \lambda_2 \in \MBB{R}$, where $\alpha \ge 2$, $0\le \beta \le 1$, and $\lambda_1, \lambda_2 \ge 0$, be constants to be decided later.

\smallskip

For $0\le i \le \FLOOR{\log_\alpha d}$, consider the subsequence $\MC{S}^\alpha_i$ defined as $$\MC{S}^\alpha_i = \BIGP{\sigma^\uparrow_d(d-\alpha^i), \sigma^\uparrow_d(d-\alpha^i+1), \ldots, \sigma^\uparrow_d(d-1)}.$$
For notational brevity, $\MC{S}^\alpha_i$ is defined to be 
a sequence of infinite zeros for $i > \FLOOR{\log_\alpha d}$.
%

%


\smallskip

We prove this lemma by showing that, with properly chosen constants $\alpha, \beta, \lambda_1$, and $\lambda_2$, for each $0 < i \le \FLOOR{\log_\alpha d}$, the total workload in $\MC{S}^\alpha_{i+1} \backslash \MC{S}^\alpha_i$ can be jointly covered by taking 
\begin{compactitem}
\item[(1)] 
$\lambda_1$ times the total density generated by $\MC{S}^\alpha_{i+1} \backslash \MC{S}^\alpha_i$ itself, and

\item[(2)]
$\lambda_2$ times the total density generated by $\MC{S}^\alpha_i \backslash \MC{S}^\alpha_{i-1}$.
\end{compactitem}

If this is true, then we get the non-negativity of the potential function $\Phi_{\lambda_1+\lambda_2}$, where the value $\lambda_1 + \lambda_2$ will be at most $5.2$ by our setting.

\smallskip

Consider the following two cases.

\begin{itemize}
\item 
For the base case, consider $\MC{S}^\alpha_0 = \BIGP{\sigma^\uparrow_d(d-1)}$. The density generated by $\MC{S}^\alpha_0$ is exactly $\sigma^\uparrow_d(d-1)$. 
By the non-decreasing property of $\sigma^\uparrow_d$, we know that the total workload in $\MC{S}^\alpha_1$ is at most $\alpha\cdot\sigma^\uparrow_d(d-1)$.
Hence, as long as $\lambda_1+\lambda_2 \ge \alpha$, the density generated by $\MC{S}^\alpha_0$ is sufficient to cover the total workload of $\MC{S}^\alpha_1$.

\smallskip

\item
For $0 < i \le \FLOOR{\log_\alpha d}$, let $\#_{i+1}$ denote the total amount of workload in $\MC{S}^\alpha_{i+1} \backslash \MC{S}^\alpha_i$. More precisely, we have
$$\#_{i+1} = \sum_{d-\alpha^{i+1} \le t < d-\alpha^i} \sigma^\uparrow_d(t).$$
Note that, from the non-decreasing property of $\sigma^\uparrow_d$, we have $0\le \#_{i+1} \le (\alpha-1)\alpha^i\cdot\sigma^\uparrow_d(d-\alpha^i)$.

\bigskip

Consider the density generated by
\begin{compactitem}
\item[\bf (i) $\MC{S}^\alpha_i \backslash \MC{S}^\alpha_{i-1}$.]
We have $\sigma^\alpha_d(t) \ge \sigma^\alpha_d(d-\alpha^i)$ for all $t \ge d-\alpha^i$. Hence,
\begin{align*}
\sum_{d-\alpha^i \le t < d-\alpha^{i-1}}\hat\varrho\BIGP{\MC{J}^\uparrow_d(t)} \: & \ge \: \sum_{d-\alpha^i \le t < d-\alpha^{i-1}}\rho\BIGP{\MC{J}^\uparrow_d(t),d-\alpha^i,d} \\
& \ge \: \sum_{1 \le k \le (\alpha-1)\alpha^{i-1}}\frac{1}{\alpha^i}\cdot k\cdot \sigma^\uparrow_d(d-\alpha^i) \\
& \ge \: \frac{1}{\alpha^i}\cdot \frac{1}{2}(\alpha-1)^2\alpha^{2i-2} \cdot\sigma^\uparrow_d(d-\alpha^i) \\
& = \: \frac{1}{2}\cdot(\alpha-1)^2\alpha^{i-2}\cdot\sigma^\uparrow_d(d-\alpha^i).
\end{align*}

\bigskip

\item[\bf (ii) $\MC{S}^\alpha_{i+1} \backslash \MC{S}^\alpha_i$.]
For this part, we derive a lower bound on the total amount of density generated by $\MC{S}^\alpha_{i+1} \backslash \MC{S}^\alpha_i$. In particular, we prove the following claim:
\begin{equation}
\sum_{d-\alpha^{i+1} \le t < d-\alpha^i}\hat\varrho\BIGP{\MC{J}^\uparrow_d(t)} \ge \frac{1}{2\cdot\alpha^{i+1}\cdot\sigma^\uparrow_d(d-\alpha^i)}\cdot\BIGP{\#_{i+1}}^2.
\label{ieq-5-2-competitive-density-lower-bound}
\end{equation}

\medskip

To see why we have this lower bound, consider each $t$ with $d-\alpha^{i+1}\le t < d-\alpha^i$, we have
$$\frac{1}{\alpha^{i+1}}\cdot\sum_{d-\alpha^{i+1} \le k \le t}\sigma^\uparrow_d(k) = \rho\BIGP{\MC{J}^\uparrow_d(t), d-\alpha^{i+1}, d} \le \hat\varrho\BIGP{\MC{J}^\uparrow_d(t)}.$$
In other words, for each $t^*$ with $d-\alpha^{i+1} \le t^* \le t$, 
$\sigma^\uparrow_d(t^*)$ implicitly contributes 
$\frac{1}{\alpha^{i+1}}\cdot\sigma^\uparrow_d(t^*)$ amount of density to $\rho\BIGP{\MC{J}^\uparrow_d(t), d-\alpha^{i+1}, d}$, which in turn serves as one lower bound to $\hat\varrho\BIGP{\MC{J}^\uparrow_d(t)}$. The total contribution of $\sigma^\uparrow_d(t^*)$ into $\sum_{d-\alpha^{i+1} \le t < d-\alpha^i}\hat\varrho\BIGP{\MC{J}^\uparrow_d(t)}$ is therefore at least $\BIGP{d-\alpha^i-t^*}\cdot\sigma^\uparrow_d(t^*)$, for all $d-\alpha^{i+1} \le t^* < d-\alpha^i$.
This is minimized as $t^*$ tends to $d-\alpha^i$. 

\smallskip

Moreover, from the non-decreasing property, we know that $\sigma^\uparrow_d(t^*)$ is at most $\sigma^\uparrow_d(d-\alpha^i)$. Therefore, when the total amount of workload in $\MC{S}^\alpha_{i+1} \backslash \MC{S}^\alpha_i$ is $\#_{i+1}$, a lower bound on the total density generated in terms of the overall contribution of each $\sigma^\uparrow_d(t^*)$, $d-\alpha^{i+1}\le t^*< d-\alpha^i$, is at least
$$\frac{1}{\alpha^{i+1}}\cdot\sigma^\uparrow_d(d-\alpha^i)\cdot\frac{1}{2}\cdot\BIGP{\frac{\#_{i+1}}{\sigma^\uparrow_d(d-\alpha^i)}}^2 = \frac{\BIGP{\#_{i+1}}^2}{2\cdot\alpha^{i+1}\cdot\sigma^\uparrow_d(d-\alpha^i)},$$
which proves our claim.
\end{compactitem}
\end{itemize}
We want to show that, with properly chosen constants $\lambda_1, \lambda_2$, and $\alpha$,
\begin{align}
\lambda_1\cdot\frac{1}{2\cdot\alpha^{i+1}}\cdot\frac{\BIGP{\#_{i+1}}^2}{\sigma^\uparrow_d(d-\alpha^i)} + \lambda_2\cdot\frac{1}{2}\cdot(\alpha-1)^2\alpha^{i-2}\cdot\sigma^\uparrow_d(d-\alpha^i) \ge \#_{i+1}.
\label{ieq-feasibility-lla}
\end{align}
Let $$\beta = \frac{\#_{i+1}}{(\alpha-1)\alpha^i\cdot\sigma^\uparrow_d(d-\alpha^i)}.$$
Note that, by the non-decreasing property and the definition of $\#_{i+1}$, we have $0\le \beta \le 1$. Then, the inequality~(\ref{ieq-feasibility-lla}) simplifies to
$$\lambda_1\cdot\frac{1}{2}\cdot\frac{\alpha-1}{\alpha}\cdot\beta^2 + \lambda_2\cdot\frac{1}{2}\cdot\frac{\alpha-1}{\alpha^2} \ge \beta.$$
Define the function $F(\beta)$ as $$F(\beta) := \lambda_1\cdot\frac{1}{2}\cdot\frac{\alpha-1}{\alpha}\cdot\beta^2 + \lambda_2\cdot\frac{1}{2}\cdot\frac{\alpha-1}{\alpha^2} - \beta.$$
Then inequality~(\ref{ieq-feasibility-lla}) holds if and only if $F(\beta) \ge 0$ for all $0\le \beta\le 1$. Note that, $F(\beta)$ is a quadratic polynomial whose global minimum occurs at $\frac{\alpha}{\lambda_1\cdot(\alpha-1)}$ by basic calculus.
Therefore, the problem reduces to the following: 
\begin{align*}
\text{Choose $\lambda_1, \lambda_2, \alpha$ to minimize $\BIGP{\lambda_1+\lambda_2}$ while} \:
\begin{cases}
\lambda_1+\lambda_2 \ge \alpha, \\
F(\beta) \ge 0, & \text{$\forall$ $0\le\beta\le 1$.}
\end{cases}
\end{align*}
With respect fixed $\alpha$, the optimal value for the above linear program can be solved.
By fine-tuning $\alpha$ and choosing $\BIGP{\lambda_1, \lambda_2} \approx \BIGP{2.6, 2.6}$, we get a near-optimal solution to the above program, and the factor is $\lambda_1+\lambda_2 \le 5.2$. This proves the lemma.
\qed
\end{proof}

\noindent
We conclude our result by the following theorem.

\begin{theorem}
The packing-via-density(5.2) algorithm computes a feasible $5.2$-competitive schedule for the machine-minimizing job scheduling problem with unit job execution time.
\end{theorem}

\begin{proof}
This theorem follows directly from the description and the lemmas given in the content.
From Lemma~\ref{lemma-reduction-equal-deadline}, Lemma~\ref{lemma-reduction-non-decreasing}, and Lemma~\ref{lemma-5-2-potential}, we know that the potential function $\Phi_{5.2}\BIGP{\MC{J}, \ell, r}$ is non-negative for any job set $\MC{J}$ and any $0\le \ell < r$. By Lemma~\ref{lemma-feasibility-characterization}, this asserts the feasibility of \emph{packing-via-density(5.2)}.

Since $\hat\varrho\BIGP{\MC{J}(t)} \le \hat\varrho\BIGP{\MC{J}}$ for all $t \ge 0$, by Lemma~\ref{lemma-opt-offline-density}, we have $\hat\varrho\BIGP{\MC{J}(t)} \le \mathrm{OPT}\BIGP{\MC{J}}$ for all $t\ge 0$. Since the algorithm \emph{packing-via-density(5.2)} uses exactly $5.2 \times \hat\varrho\BIGP{\MC{J}(t)}$ number of machines, 
the resulting schedule 
is $5.2$-competitive.
\qed
\end{proof}

\section{Conclusion} \label{section_conclusion}

This paper presents online algorithms and competitive analysis for 
the machine-minimizing job scheduling problem with unit jobs. We disprove a false claim made by a previous paper regarding a further restricted case. We also provide a lower bound on the competitive factor of any online algorithm for this problem.



\bibliographystyle{siam}
\bibliography{online-bin-packing}

%
%
%


\end{document}